\documentclass[12pt]{amsart}

\usepackage{amssymb,latexsym,amsmath,amsthm,amsfonts,graphicx}
\numberwithin{equation}{section}
\newtheorem{thm}{Theorem}[section]
\newtheorem{conj}[thm]{Conjecture}
\newtheorem{lemma}[thm]{Lemma}
\newtheorem*{remark}{Remark}

\begin{document}

\title[Isothermal lensing]{Transcendental Harmonic Mappings and Gravitational Lensing by Isothermal Galaxies}


\author[D. Khavinson]{Dmitry Khavinson}
\address{Department of Mathematics and Statistics\\
University of South Florida\\ 4202 E. Fowler Ave., PHY114\\ Tampa FL
33617} 
\email{dkhavins@cas.usf.edu}

\author[E. Lundberg]{Erik Lundberg}
\address{Department of Mathematics and Statistics\\
University of South Florida\\ 4202 E. Fowler Ave., PHY114\\ Tampa FL
33617}
\email{elundber@mail.usf.edu}

\footnotetext[1]{Both authors gratefully acknowledge partial support from the National Science Foundation.}

\date{2009}

\begin{abstract}
Using the Schwarz function of an ellipse, it was recently shown that galaxies with density constant on confocal ellipses can produce at most four ``bright'' images of a single source.  The more physically interesting example of an isothermal galaxy has density that is constant on \emph{homothetic} ellipses.  In that case bright images can be seen to correspond to zeros of a certain transcendental harmonic mapping.  We use complex dynamics to give an upper bound on the total number of such zeros.
\end{abstract}

\maketitle

\section{Introduction}

In this paper, we obtain an upper bound for the number of solutions of the equation

\begin{equation}
	\arcsin\left(\frac{k}{\bar{z}+\bar{w}}\right) = z,
	\label{lens2}
\end{equation}

where $w$ is a complex parameter, and $k$ is a real parameter.

Our motivation for doing so is that solutions of (\ref{lens2}) in fact correspond to virtual images observed when the light from a distant source passes near an isothermal, ellipsoidal galaxy.  Indeed, using the complex formulation of the thin-lens approximation (\cite{St}), the lensing equation is calculated by finding the Cauchy transform of the mass distribution projected to the ``lens plane''.  This was carried out in \cite{K} with the following result (we also sketch the derivation in the last section for the reader's convenience).

\begin{equation}
	\mathcal{C} \arcsin\left(\frac{c}{\bar{\zeta}}\right)+ \omega = \zeta
	\label{lens1}
\end{equation}

Here, we take the principal branch of $\arcsin$, $\mathcal{C}$ and $c$ are real constants depending on the elliptical projection of the galaxy onto the lens plane, and $\omega$ is the position of the source (projected to the lens plane), and values of $\zeta$ which satisfy (\ref{lens1}) give positions of the observed images.  Changing variables to $z = \frac{\zeta-\omega}{\mathcal{C}}$, $w=\omega/\mathcal{C}$, and $k=c/\mathcal{C}$ puts (\ref{lens1}) into the form of equation (\ref{lens2}) while preserving the number of solutions.

We should mention that the anti-analytic potential in the lensing equation considered here (and, also,
in \cite{BE} and \cite{K}) differs from the potential in the lensing equation in the  model often used by astrophysicists (see \cite{KMW} and the references therein), where the projected mass density is supported in the entire complex plane.  Both models use the ``isothermal'' density proportional to $1/t$ on ellipses $\{x^2/a^2+y^2/b^2=t^2\}$ ($a$ and $b$ fixed).  The model considered here that yields equation (\ref{lens1}) assumes that the density is zero for all $t$ greater than some value (see Appendix).  Letting the density have infinite support assumes that the galaxy has infinite mass and fills the universe, yet it is the simplest way to avoid giving the galaxy a ``sharp edge'' and astronomers have found that the model behaves reasonably in the region where the lensed images occur.  We consider the model with physically realistic compact support but less realistic ``sharp edge'' for a mathematical reason: in that setting, lensed images described by solutions of equation (\ref{lens2}) correspond to zeros of a \emph{harmonic} function (We note that models with ``sharp edges'' have been  considered by astrophysicists as well, cf. the recent preprints  \cite{PRT1} and \cite{PRT2}).

For gravitational lenses consisting of $n$ point masses, Mao, Petters, and Witt \cite{MPW} suggested (1997) that the bound for the number of images was linear in $n$ (Bezout's theorem provides a bound quadratic in $n$). Rhie refined this in 2001, conjecturing that a gravitational lens consisting of $n$ point masses cannot create more than $5n - 5$ images of a given source \cite{R1}. In 2003, she constructed point-mass configurations for which these bounds are attained \cite{R2}.  The first author and G. Neumann \cite{K-N} settled her conjecture by giving a bound of $5n-5$ zeros for harmonic mappings of the form $r(z)-\bar{z}$, where $r(z)$ is rational of deg $n > 1$. (See \cite{K-N2} for the exposition and further details.)  Solutions of (\ref{lens2}) are zeros of a \emph{transcendental} harmonic function, so extending the techniques used in \cite{K-N} will require some care (a priori, it is not even clear that the number of zeros is finite, cf. \cite{B-G}, \cite{Kraj}).  Still, our approach draws on the same two main results: (i) the argument principle generalized to harmonic functions and (ii) the Fatou theorem from complex dynamics regarding the attraction of critical points.  In the next section, we will formulate (i).  (ii) will have to be modified for our purposes, so ideas from complex dynamics are worked from scratch into the proof of Lemma \ref{sp} in the third section.  

{\bf Acknowledgement:} We would like to thank Walter Bergweiler, Alex Eremenko, and Charles R. Keeton for stimulating discussions and for sharing some of their unpublished work with us.

\section{Preliminaries: The Argument Principle}

In order to state the generalized argument principle (see \cite{Duren} for a complete exposition and proof), we need to define the order of a zero or pole of a harmonic function.  A harmonic function $h = f + \bar{g}$, where $f$ and $g$ are analytic functions, is called \emph{sense-preserving} at $z_0$ if the Jacobian $Jh(z) = |f'(z)|^2-|g'(z)|^2 > 0$ for every $z$ in some punctured neighborhood of $z_0$. We also say that $h$ is \emph{sense-reversing} if $\bar{h}$ is sense-preserving at $z_0$. If $h$ is neither sense-preserving nor sense-reversing at $z_0$, then $z_0$ is called singular and necessarily (but not sufficiently) $Jh(z_0) = 0$, cf. \cite{Duren}, Ch. 2.  The \emph{order} of a non-singular zero is given by $\frac{1}{2\pi}\Delta_C \arg h(z)$, where $C$ is a sufficiently small circle around the zero.  The order is positive if $h$ is sense-preserving at the zero and negative if $h$ is sense-reversing.

Suppose $h$ is harmonic in a punctured neighborhood of $z_0$. We will refer to $z_0$ as a pole of $h$ if $h(z) \rightarrow \infty$ as $z \rightarrow z_0$. Let $C$ be an oriented closed curve that contains neither zeros nor poles of $h$. The notation $\frac{1}{2\pi}\Delta_C \arg h(z)$ denotes the increment in the argument of $h(z)$ along $C$.
Following \cite{Su}, the \emph{order} of a pole of $h$ is given by $-\frac{1}{2\pi}\Delta_C \arg h(z)$, where $C$ is a sufficiently small circle around the pole. We note that if $h$ is sense-reversing on a sufficiently small circle around the pole, then the order of the pole will be negative. We will use the following version of the argument principle which is taken from \cite{Su}:

\begin{thm}\label{AP}
Let $F$ be harmonic, except for a finite number of poles, in a Jordon domain $D$.  Let $C$ be a curve contained in $D$ not passing through a pole or a zero, and let $R$ be the open, bounded region surrounded by $C$. Suppose $F$ has no singular zeros in $R$ and let $N$ be the sum of the orders of the zeros of $F$ in $R$.  Let $P$ be the sum of the orders of the poles of $F$ in $R$.  Then $\Delta_C \arg F(z) = 2 \pi (N-P)$.
\end{thm}

\section{An Upper Bound for the Number of Images}

\begin{lemma}\label{R}
The solutions of equation (\ref{lens2}) are all contained in a rectangle, $R:=\{|Re(z)|\leq \pi/2,|Im(z)|\leq M\}$, where $M$ is sufficiently large.
\end{lemma}

\begin{proof}
The requirement that $|Re(z)|\leq \pi/2$ is immediate since this strip is the image of $\mathbb{C}$ under the principal branch of $\arcsin$.  To see that there exists an $M$ such that solutions of (\ref{lens2}) satisfy $|Im(z)| \leq M$, take $\sin$ of both sides.  This leads to

\begin{equation}
	\frac{k}{\overline{\sin(z)}}=z+w.
	\label{lens3}
\end{equation}

We consider the modulus of each side of (\ref{lens3}) for $z=x+iy$ with large values of $|y|$.  Recall, $\sin(x+iy)=\sin(x)\cosh(y)+i\cos(x)\sinh(y)$.  As $y \rightarrow \pm \infty$, $|\frac{k}{\sin(x+iy)}|= k/\sqrt{\sin^2x \cosh^2y+\cos^2x \sinh^2y} \rightarrow 0$, uniformly in $x$.  On the other hand, $|z+w| \rightarrow \infty$.

\end{proof}

\noindent{\bf Remark:} By this lemma, we can bound the number of solutions of (\ref{lens2}) by bounding the number of zeros of $F(z):=z+w-\frac{k}{\overline{\sin(z)}}$ in the rectangle, $R$.  Let us calculate the increment of the argument of $F(z)$ when $\partial R$ is traced counterclockwise.

\begin{lemma}\label{arg}
$\Delta_{\partial R}\arg F(z) \geq -2\pi$, where $F(z):=z+w-\frac{k}{\overline{\sin(z)}}$.
\end{lemma}

\begin{proof}
Consider the four links $V_{\pm}(t)=\pm\frac{\pi}{2}\pm i(-M+t)$, $0\leq t \leq 2M$, $H_{\pm}(t)= \pm(\frac{\pi}{2}-t) \pm iM$, $0\leq t \leq \pi$, which trace the right, left, top, and bottom edges, respectively.  We need to determine the effect of the term $-\frac{k}{\overline{\sin(z)}}$.  Without this term, $F(z)$ is just translation $z \rightarrow z + w$, and in that case $F(V_{\pm}(t))$ and $F(H_{\pm}(t))$ trace the edges of the translated rectangle.  

By choosing $M$ large enough in the previous lemma, we can neglect the term $-\frac{k}{\overline{\sin(z)}}$ on the top and bottom edges.  On the right edge, $\frac{k}{\overline{\sin(V_{+}(t))}}=\frac{k}{\cosh(-M+t)}$ is pure real and increases monotonically from a small value at $t=0$ to the value $k$ at $t=M$.  On the interval $M \leq t \leq 2M$, $\frac{k}{\cosh(-M+t)}$ \emph{decreases} monotonically from $k$ at $t=M$ back to the original value at $t=2M$.  Similarly, on the left edge, $\frac{k}{\overline{\sin(V_{-}(t))}}=-\frac{k}{\cosh(-M+t)}>-k$.  Thus, the effect of the term $-\frac{k}{\overline{\sin(z)}}$ is to bend the left and right sides of the translated rectangle inward, so that they cross each other if and only if $k>\frac{\pi}{2}$ (compare the two images in figure \ref{both}).

If $k<\pi/2$, then the images of the left and right edges do not intersect, and either $\Delta_{\partial R}\arg F(z) = 2\pi$, or, if $F(\partial R)$ does not surround the origin, $\Delta_{\partial R}\arg F(z) = 0$.  See the left image in figure \ref{both}.

\begin{figure}[h]
	\includegraphics[scale=0.5]{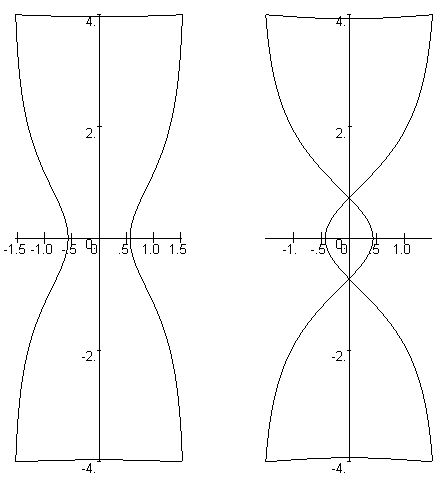}
	\caption{The image of $\partial R$ under $F(z)$ with $w=0$ and two choices for the value of $k$.  In the left, $k=1<\pi/2$.  In this case, $\Delta_{\partial R}\arg F(z) = 1$.  If we set $w$ to, 	say, $1$ then we have $\Delta_{\partial R}\arg F(z) = 0$. In the right, $k=2>\pi/2$, and $\Delta_{\partial R}\arg F(z) = -1$.  If we set $w$ to, say, $1$ or $i$ then we have $\Delta_{\partial R}\arg F(z) = 0$ or $1$, respectively.}
	\label{both}
\end{figure}

If $k>\pi/2$, then the images of the left and right edges intersect exactly twice.  In this case, there is a third possibility in which $\Delta_{\partial R}\arg F(z)$
$ = -2\pi$.  See the right image in figure \ref{both}.

\end{proof}

Define the anti-analytic function $f(z):=\frac{k}{\overline{\sin(z)}}-w$ whose fixed points coincide with the zeros of $F(z)$.  Notice that sense-preserving zeros of $F(z)$ coincide with attracting fixed points of $f(z)$ (see \cite{DHL} and \cite{C-G}).  Indeed, a fixed point $z_0$ is attracting when $|f'(z_0)|<1$ which, if satisfied, holds in a neighborhood of $z_0$ so that $1-|f'(z)|>0$ near $z_0$ which is the condition required for $F(z)$ to be sense-preserving at $z_0$.  We use complex dynamics to bound the number of attracting fixed points of $f(z)$.  The version of the Fatou theorem found in most textbooks on complex dynamics such as \cite{C-G} falls short, since $f$ has infinitely many critical points, and $f$ is not analytic but rather anti-analytic.  The extensions of the Fatou theorem discussed in, e.g., the survey \cite{Berg} are almost sufficient to cover our situation, but $f$ is still assumed to be analytic.  For the reader's convenience, we give an independent direct proof of the next lemma.

\begin{lemma}\label{sp}
The number of sense-preserving zeros, $n_+$, is at most 3.
\end{lemma}

\begin{proof}
Suppose $z_0$ is a sense preserving zero.  Then $|f'(z_0)| = \gamma < 1$, and there is a neighborhood, $U_1$, of $z_0$ which is contracted to $z_0$ under iterations of $f(z)$.  Consider $f^{-1}(\zeta)=\arcsin(\frac{k}{\overline{\zeta+w}})$.  Recall that $\arcsin$ has a single-valued, injective branch in any simply connected domain in $\mathbb{C} \setminus \{\pm 1, \infty \}$.  Since $\frac{k}{\overline{\zeta+w}}$ injectively maps $\mathbb{C} \setminus \{-w, -w \pm k \}$ into $\mathbb{C} \setminus \{\pm 1, \infty \}$, we conclude $f^{-1}$ has a single-valued, injective branch in any simply connected domain which omits $\{-w, -w \pm k \}$.  

Now we consider the algorithm described in the proof of Fatou's theorem.  Suppose $U_1$ omits $\{-w, -w \pm k \}$.  Then choose the single-valued branch of $f^{-1}$ so that $f^{-1}(z_0)=z_0$.  Then $U_2:=f^{-1}(U_1)$ is a simply connected domain, since $f^{-1}$ is analytic and injective.  Therefore, if $U_2$ omits $\{-w, -w \pm k \}$ then $f^{-1}$ has a single-valued extension to $U_2$.  We can proceed inductively, defining $U_{n+1}:=f^{-1}(U_n)$, provided $U_n$ omits $\{-w, -w \pm k \}$.  

If the algorithm does not stop, then it produces an infinite family, $\{f^{-2n}\}_{n=0}^{\infty}$, of functions analytic on the domain $U_1$ (odd iterates of $f^{-1}$ are anti-analytic).  This family omits $\{-w, -w \pm k \}$ and is therefore a normal family (by Montel's theorem).  This contradicts the divergence of $\frac{d}{d\zeta}f^{-n}(\zeta)|_{\zeta=z_0}$ which can be seen by repeated application of the chain rule along with the fact that $\frac{d}{d\zeta}f^{-1}(\zeta)|_{\zeta=z_0}$ has modulus $1/\gamma > 1$.  Thus, at least one of the three points $-w, -w \pm k $ is attracted to $z_0$ under iteration of $f(z)$, and the lemma follows.
\end{proof}

\begin{thm}\label{bound}
The number of solutions to (\ref{lens2}) is bounded by 8.
\end{thm}

\begin{proof}
By the remark following Lemma \ref{R}, the total number of solution to (\ref{lens2}) equals the total number of zeros of $F(z)$ in $R$.  Recall that $F(z)$ is called ``regular'' if it is free of singular zeros (see \cite{K-S} and \cite{DHL}).  Suppose for the moment that $F(z)$ is regular.  Then, the total number of zeros of $F(z)$ in $R$ is $N = n_+ + n_-$, where $n_+$ and $n_-$ count, respectively, the sense-preserving and sense-reversing zeros of $F(z)$ in $R$.  By Lemma \ref{arg} and Theorem \ref{AP}, $-1 \leq N - P$.  By Lemma \ref{sp} and the fact that $F(z)$ has one sense-reversing pole in $R$, this becomes $-1 \leq 3 - n_- + 1$, so that $n_- \leq 5$.  Thus, $N = n_+ + n_- \leq 8$.

Fix $k$.  There is a dense set of parameters $w$ for which $F(z)$ is regular.  Indeed, consider the image of $\{z:|\frac{d}{dz}(\frac{k}{\sin(z)})|=1\}$ under $z-\frac{k}{\overline{\sin(z)}}$.  This set has empty interior, and if $w$ is in its complement, $F(z)$ is free of singular zeros.

Now suppose $F(z)$ is not regular.  If $F(z)$ is perturbed by a sufficiently small constant, the number of zeros is not changed.  By the preceding, we can choose an arbitrarily small constant so that the perturbation is regular.  Thus, the bound, $N \leq 8$, holds for all $F(z)$.
\end{proof}

\section{Concluding Remarks}

\begin{figure}[h]
\includegraphics[scale=.5]{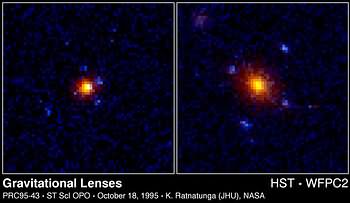}
\caption{Four images of a light source behind an elliptical galaxy. (Credit: NASA, Kavan Ratnatunga, Johns Hopkins Univ.)}
\label{4images}
\end{figure}

So far, astronomers have observed only up to 5 images (4 bright + 1 dim) produced by an elliptical lens (see figure \ref{4images}).  In \cite{KMW} there have been constructed explicit models (depending on the semiaxes of the ellipse) having 9 images (8 bright + 1 dim) but only in the presence of a shear, i.e. a (linear) gravitational pull from infinity (a term $\gamma \bar{z}$ added to equation (\ref{lens2})).  So far, we have not been able to obtain a universal bound in the presence of a shear that is similar to Theorem \ref{bound}.  It seemed, based on NASA observations, natural to conjecture that, in the absence of shear, there can be at most 4 bright images.  Yet, recently W. Bergweiler and A. Eremenko generated an example with 6 bright images \cite{BE}. With their kind permission, we include their example (see figure \ref{AlexWalter}). Based on their breakthrough and the investigation carried out in \cite{KMW}, we conjecture the following.

\begin{conj}
(i) The number of bright images lensed by an isothermal elliptical galaxy without shear is at most 6.
(ii) In the presence of a shear, the number of bright images is at most 8.
\end{conj}

We caution the reader that in [8] the mass density was assumed to
be extended all the way to infinity, so the lensing potential in [8] was
different from the one we consider here (and in [3]).

\begin{figure}[h]
\includegraphics[scale=.6]{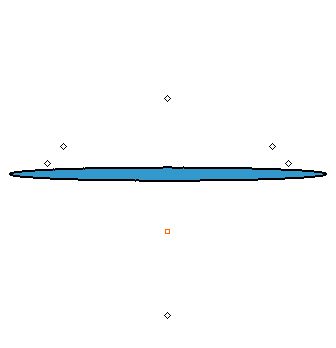}
\caption{Equation \ref{lens2} has 6 solutions when k=1.92 and w=-.67i.  Choosing a = 1, b = .041 , and M = 2 (see Appendix) leads to k = 1.92 and gives the picture of the six images shown here in the $\zeta$-plane along with the galaxy's elliptical silhouette and the source (plotted as box).)}
\label{AlexWalter}
\end{figure}

{\bf Added in proof}:  After this paper was submitted, Walter Bergweiller and Alex Eremenko succeded in proving  part (i) of the conjecture.

\section{Appendix: derivation of the complex lensing equation for the isothermal elliptical galaxy}

Suppose that light from a distant source star is distorted as it passes by an intermediate, continuous distribution of mass which does not deviate too far from being contained in a common plane (the ``lens plane'') perpendicular to our line of sight.  Let $\mu(z)$ denote the projected mass density.  Then basic results from General Relativity combined with Geometric Optics (see \cite{St}) lead to the following lensing equation relating the position of the source (projected to the lens plane) $w$ to the positions of lensed images $z$.

\begin{equation}
	\label{lens}
	z = \int_{\Omega}{\frac{\mu(\zeta)dA(\zeta)}{\bar{\zeta}-\bar{z}}}+w
\end{equation}

Consider, first, the case when the projected density $\mu(z)=D$ is constant and supported on $\Omega:=\{\frac{x^2}{a^2}+\frac{y^2}{b^2} \leq 1, a> b>0\}$, an ellipse.  Then equation (\ref{lens}) becomes

$$z = \int_{\Omega}{\frac{D dA(\zeta)}{\bar{\zeta}-\bar{z}}}+w.$$

By the complex Green's formula, for $z$ outside $\Omega$ (i.e., for ``bright'' images), this becomes

$$z = \frac{D}{2i}\int_{\partial \Omega}{\frac{\zeta d\bar{\zeta}}{\bar{\zeta}-\bar{z}}}+w.$$

The Schwarz function (by definition, analytic and $= \bar{\zeta}$ on $\partial \Omega$) for the ellipse equals ($c^2=a^2-b^2$):
$$S(\zeta) = \frac{a^2+b^2}{c^2}\zeta-\frac{2ab}{c^2}(\sqrt{\zeta^2-c^2})$$
$$=\frac{a^2+b^2-2ab}{c^2}\zeta+\frac{2ab}{c^2}(\zeta-\sqrt{\zeta^2-c^2})$$
$$=S_1(\zeta)+S_2(\zeta)$$
where $S_1$ is analytic in $\overline{\Omega}$, and $S_2$ is analytic outside $\Omega$ and $S_2(\infty) = 0$.  Since $z$ is outside $\Omega$, combining this with Cauchy's formula gives
$$\frac{1}{2i}\int_{\partial \Omega}{\frac{\overline{S(\zeta)} d\bar{\zeta}}{\bar{\zeta}-\bar{z}}}+w=\pi \frac{2ab}{c^2}D(\bar{z}-\sqrt{\bar{z}^2-c^2})+w$$
for the right-hand-side of the lensing equation.

Next consider the case of ``isothermal'' density supported on $\Omega$, $\mu = M/t$ on $\partial \Omega_t$, $\Omega_t:=t\Omega=\{\frac{x^2}{a^2}+\frac{y^2}{b^2} \leq t^2\}, t<1$, and $M$ a constant.

Then the Cauchy potential term in the lensing equation (\ref{lens}) becomes 

\begin{equation}
	\label{CP}
\int_{\Omega}{\frac{\mu(z)}{\bar{\zeta}-\bar{z}}dA(\zeta)}=\int_{0}^{1}{\frac{M}{t}\left[\frac{d}{dt}\int_{\Omega_t}{\frac{dA(\zeta)}{\bar{\zeta}-\bar{z}}}\right]dt}
\end{equation}

For the inside integral, we see that $\int_{\Omega_t}{\frac{dA(\zeta)}{\bar{\zeta}-\bar{z}}}=t^2\int_{\Omega}{\frac{dA(\zeta)}{t\bar{\zeta}-\bar{z}}}=t\int_{\Omega}{\frac{dA(\zeta)}{\bar{\zeta}-\bar{z}/t}}$ which according to our previous calculation is $C_0(\bar{z}-\sqrt{\bar{z}^2-c^2t^2})$, where the constant $C_0$ depends only on $\Omega$.  Now the $t$-derivative of this is $C_0 \frac{t}{\sqrt{\bar{z}^2-c^2t^2}}$. Thus (\ref{CP}) becomes $MC_0\int_{0}^{1}{\frac{dt}{\sqrt{\bar{z}^2-c^2t^2}}}$.

Finally, we arrive at (\ref{lens1}), the lensing equation for the isothermal elliptical galaxy,
$$z = \mathcal{C}\arcsin\left(\frac{c}{\bar{z}}\right) + w,$$
where $\mathcal{C}=\frac{2\pi ab}{c}M$.

\bibliographystyle{amsplain}

\end{document}